\documentclass[10pt, conference, compsocconf]{IEEEtran}
\IEEEoverridecommandlockouts

\ifCLASSINFOpdf
  \usepackage[pdftex]{graphicx}
\else
\fi
\usepackage[cmex10]{amsmath}
\usepackage{amssymb}
\usepackage{amsfonts}

\usepackage{algorithmic}
\usepackage[tight,footnotesize]{subfigure}

\usepackage{graphicx}
\usepackage{bm}
\usepackage{booktabs}
\usepackage{multirow}

\newtheorem{lemma}{Lemma}
\newtheorem{proof}{Proof}

\begin{document}
%
\title{Learning Node Representations from Noisy Graph Structures}


\author{\IEEEauthorblockN{Junshan Wang\IEEEauthorrefmark{1},
Ziyao Li\IEEEauthorrefmark{2},
Qingqing Long\IEEEauthorrefmark{1}, 
Weiyu Zhang\IEEEauthorrefmark{1},
Guojie Song\IEEEauthorrefmark{1}\textsuperscript{$\S$} and
Chuan Shi\IEEEauthorrefmark{3}} \thanks{\textsuperscript{$\S$} Corresponding author. 

We are very thankful to Lun Du for his important and helpful suggestions for our model. 
This work was supported by the National Natural Science Foundation of China (Grant No. 61876006).}
\IEEEauthorblockA{\IEEEauthorrefmark{1}Key Laboratory of Machine Perception, Ministry of Education, Peking University, Beijing, China}
\IEEEauthorblockA{\IEEEauthorrefmark{2}Center for Data Science, Peking University, Beijing, China \\
Email: \{wangjunshan, leeeezy, qingqinglong, zhangdavid, gjsong\}@pku.edu.cn}
\IEEEauthorblockA{\IEEEauthorrefmark{3}School of Computer Science, Beijing University of Posts and Telecommunications, Beijing, China \\
Email: shichuan@bupt.edu.cn}}


%


\maketitle

\begin{abstract}
    Learning low-dimensional representations on graphs has proved to be effective in various downstream tasks. However, noises prevail in real-world networks, which compromise networks to a large extent in that edges in networks propagate noises through the whole network instead of only the node itself. While existing methods tend to focus on preserving structural properties, the robustness of the learned representations against noises is generally ignored. 

    In this paper, we propose a novel framework to learn noise-free node representations and eliminate noises simultaneously. Since noises are often unknown on real graphs, we design two generators, namely a graph generator and a noise generator, to identify normal structures and noises in an unsupervised setting. On the one hand, the graph generator serves as a unified scheme to incorporate any useful graph prior knowledge to generate normal structures. We illustrate the generative process with community structures and power-law degree distributions as examples. On the other hand, the noise generator generates graph noises not only satisfying some fundamental properties but also in an adaptive way. Thus, real noises with arbitrary distributions can be handled successfully. Finally, in order to eliminate noises and obtain noise-free node representations, two generators need to be optimized jointly, and through maximum likelihood estimation, we equivalently convert the model into imposing different regularization constraints on the true graph and noises respectively. Our model is evaluated on both real-world and synthetic data. It outperforms other strong baselines for node classification and graph reconstruction tasks, demonstrating its ability to eliminate graph noises. 
\end{abstract}


%
\IEEEpeerreviewmaketitle

\section{Introduction}

{\bf Node representation learning} (NRL) (or Network Embedding) learns low-dimensional representations of nodes on complex graphs. 
While traditional models like DeepWalk \cite{perozzi2014deepwalk}, LINE \cite{tang2015line}, GCN \cite{kipf2016semi} and GraphSAGE \cite{hamilton2017inductive} prove to be effective on node representation tasks, a common drawback of previous works is that noises on graphs are rarely stressed. 

In the real world, all data may be contaminated by noises, and networks are no exception. For example, spammers connect to nodes without distinction on social networks, leading to numerous edge noises. The process of collecting network data may also suffer from noises such as missing edges. 
Noises will propagate through the entire graphs so that representations of any nodes may be affected. Meanwhile, due to the complex nature of noises, it is hard to exactly identify noises and hence adopt supervised methods to learn robust representations. 
However, neglecting noises may lead to inferior representations and unreliable predictions in downstream applications. For example, the high connectivity of spammers could cheat the recommendation system. 

Indeed, noise detection algorithm \cite{akoglu2015graph,akoglu2010oddball} can be leveraged but they are generally designed for specific patterns and cannot generalize to flexible notions of noises. 
Besides, there exist some approaches to enhance the robustness of representations on graphs. A great portion of them adopt adversarial training methods to encourage the global smoothness of the embedding distribution such as AIDW \cite{dai2018adversarial}, ARGA \cite{pan2018adversarially}, GraphGAN \cite{wang2018graphgan} and AdvT \cite{dai2019adversarial}. The rest of them handle the uncertainty of graph structures with the help of Bayesian methods \cite{zhang2019bayesian} or distributions \cite{zhu2019robust}. 
However, since all of them work in vector space, the comprehensibility of real noises is lacking, and hence prior knowledge cannot be incorporated to help detect noises. It leaves the necessity for a more explainable and flexible network embedding framework that handles noises in network space where properties of network structures and noises can be modeled. 

Recently, generative models have gained increasing popularity in NRL and we resort for them to address the problem of noise-handling network embedding for two reasons. On one hand, noises are generally unseen, which makes it hard for discriminative models to handle. But generative models can be designed to reconstruct the pristine graph and noises separately, thereby eliminating noises from the graph. On the other hand, prior knowledge of networks sheds light on how the networks tend to be, which can be easily incorporated in generative models. Yet the following challenges still linger:
\textbf{(1). Diverse and complex graph prior knowledge.} Incorporating prior knowledge helps to achieve a more reasonable generative process so that noises can be identified better. However, a wide range of complex prior knowledge about real networks prevails, imposing challenges on designing a universal scheme capable of integrating all of them. 
\textbf{(2). Noises with arbitrary distributions.} There are numerous types of noises on graphs, none of which possess clear definitions, making it hard to design a fixed process for real noises. Alternatively, adaptive processes for noises should be designed to handle noises with arbitrary distributions.  

In this paper, we propose a generative framework to learn noise-free node representations on graphs in an unsupervised setting. Two generative processes are designed for normal structures and noise edges respectively. 
Specifically, the graph generator is a universal scheme that is capable of incorporating any useful graph prior knowledge to facilitate the generation of normal structures.  
The noise generator is based on some fundamental properties of noises on real graphs and can eliminate noises with arbitrary distributions in an adaptive generative way.
We then propose a joint learning algorithm to optimize two generators via maximum likelihood estimation with the real graph as observations. Our model is equivalently converted into imposing specific regularization terms on the normal structures and noises. 
We evaluate our model on different datasets. Experimental results show that our model can not only obtain better performance for node classification, but also reconstruct the ground-truth graph and noises. 

To summarize, we make the following contributions.
\begin{itemize}
    \item We propose a robust network embedding model capable of eliminating noises in an unsupervised manner. Distinct generative processes for the graph and noises are designed and optimized jointly. 
    \item We design a universal scheme to incorporate different prior knowledge of graphs, and propose an adaptive method to handle noises with arbitrary distributions. 
    \item We conduct experiments for several tasks and results show that our model has the best performance of learning robust node representations.
\end{itemize}

\section{Denoising Network Embedding Model}

\begin{figure*}[h]
\centering
\includegraphics[width=1.0\textwidth]{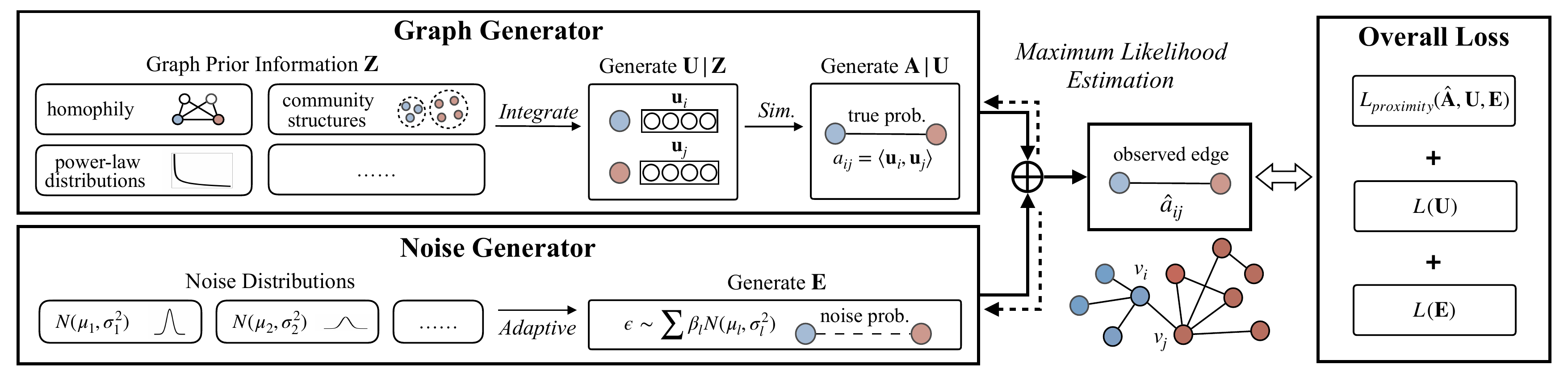}
\caption{The framework of our proposed Denoising Network Embedding Model. }
\label{fig:framework}
\end{figure*}

\subsection{Overview}
Let $\hat{G} = (V, E)$ be an undirected and unweighted graph with $n$ nodes observed from the real data. Let $\mathbf{\hat{A}}=\{\hat{a}_{ij}\}_{n \times n} \in \{0,1\}_{n \times n}$ be its adjacency matrix, where $\hat{a}_{ij} = 1$ if nodes $v_i$ and $v_j$ are linked. We assume that the observed graph is not ideal and may be corrupted by noises, which act upon the observed adjacency matrix $\hat{\mathbf{A}}$. 
We denote $\mathbf{A}=\{a_{ij}\}_{n \times n}$ as the adjacency matrix of the true graph and $\mathbf{E}=\{\epsilon_{ij}\}_{n \times n}$ as noises on the graph which act upon $a_{ij}$ as perturbations. 
Formally, the relationship among the observed graph $\mathbf{\hat{A}}$, true graph $\mathbf{A}$ and noises $\mathbf{E}$ is
\begin{equation}
    \mathbf{\hat{A}} = \mathbf{A} + \mathbf{E}.
\label{equ:sum}
\end{equation}
Notice that in this paper, we discuss the most fundamental case where noises refer to perturbations on edges. It can be generalized to more universal scenarios, such as directed and weighted graphs by following \cite{tang2015line}. And we leave other sorts of noises like noises on nodes or attributes for future work. 


In this work, we propose an NRL model to learn robust node representations and eliminate noises simultaneously as shown in Figure \ref{fig:framework}. Unlike traditional methods that learn node representations $\mathbf{U}=\{\mathbf{u}_i \in \mathbb{R}^m\}_n$ on $\hat{\mathbf{A}}$ corrupted by noises $\mathbf{E}$, our proposed model learns representations $\mathbf{U}$ that reflect $\mathbf{A}$. 
As both true data and noises on graphs are unknown, an unsupervised approach is required. We propose a generative framework consisting of a \textbf{Graph Generator} and a \textbf{Noise Generator}, which are designed according to the distinct patterns of normal edges and noise edges and will be introduced in Section \ref{sec:graph_generator} and Section \ref{sec:noise_generator} respectively. 
Then we propose a \textbf{Joint Learning Algorithm} to optimize two generators in Section \ref{sec:learning}, such that estimates of $\mathbf{A}, \mathbf{U}, \mathbf{E}$ can be obtained with good flexibility. 

\subsection{Graph Generator with Diverse Priors}
\label{sec:graph_generator}


Real-world graphs generally possess heuristic knowledge, such as homophily, community structures, hierarchical structures, degree distributions, etc. Incorporating these knowledge can contribute to a better generative process and further help eliminate noises. However, it is not easy to integrate different prior knowledge into the same generation scheme. 
In this section, we first introduce the basic scheme of the graph generator. The key of the graph generator is the generation strategies for representations $\mathbf{U}$ and the probabilities $\mathbf{A}$. 
Then our graph generator can be extended to integrate any useful prior knowledge according to the specific properties of a graph, providing stronger denoising ability and more flexibility for the model. We take two prevailing prior knowledge as examples.

We first introduce the basic idea of the graph generator. Formation of normal edges obeys some universal patterns and we use one of the most straightforward intuitions of homophily, i.e. two nodes are likely to be connected if they are ``similar" to each other. We draw support from node representations $\mathbf{U}$, in that they represent properties of nodes. For node $v_i$, we propose that its representation follows the Gaussian distribution $\mathbf{u}_i \sim \mathcal{N}(\mathbf{0}, \sigma_u^2 \mathbf{I})$, where $\sigma_u$ is a hyper-parameter. 
Given the generated representations of nodes, we believe that those nodes whose representations are close are inclined to draw an edge as
\begin{equation}
   a_{ij} = \sigma(\langle \mathbf{u}_i, \mathbf{u}_j \rangle),
\label{equ:prob}
\end{equation}
where $\langle \mathbf{x}, \mathbf{y} \rangle = - \| \mathbf{x} - \mathbf{y} \|^2$ denote the proximity between two vectors and $\sigma(x)$ is the sigmoid function.

Notice that $\mathbf{u}_i$ can be generated based on more complex prior distributions and hence the generation of $a_{ij}$ may also be adjusted. It will be discussed with two examples in the following, and furthermore, more prior knowledge of graphs can be integrated by simply following these examples.

\subsubsection{Community structures}

Community structures are common in the real world, which provide insight towards how the network is formed in that intra-community links are more likely than inter-community links. By contrast, such pattern is generally not observed for noises, and hence, prior knowledge on community structures facilitates generation of true graphs and distinguishes them from noises.

Suppose that there are $K$ communities in graph $\mathbf{A}$ and each node $v_i$ belongs to a community. 
Each community in graph $\mathbf{A}$ has a representation $\mathbf{c}_k$ that follows a Gaussian distribution $\mathcal{N}(\mathbf{0}, \sigma_w^2 \mathbf{I})$. For each node $v_i$, the first $m_1$ dimensions of its representations denote its local structure while the last part with $m_2=(m-m_1)$ dimensions denote its community information. Accordingly, the first $d$ dimensions of its representations follow the Gaussian distribution as previously stated, while the rest are generated from community representations by a Gaussian mixture formation $\mathbf{u}_{i, m_1:m} \sim \sum_k \gamma_{ik} \mathcal{N}(\mathbf{c}_k, \sigma_c^2 \mathbf{I})$, where $\sigma_c$ is hyper-parameter and $\gamma_{ik}$ denotes community membership that satisfies $\sum_k \gamma_{ik} = 1$.

When generating a normal edge $a_{ij}$, not only the local structure but also the community information is considered. Two nodes are inclined to form a link if they are in the same community. We hence calculate $a_{ij}$ in Equation \ref{equ:prob} as
\begin{equation}
a_{ij} = \sigma( \langle \mathbf{u}_{i,0:m_1}, \mathbf{u}_{j,0:m_1} \rangle + \langle \mathbf{u}_{i, m_1:m}, \mathbf{u}_{j,m_1:m} \rangle ).
\end{equation}

\subsubsection{Power-law}

Node degrees follow a power-law distribution on many real-world graphs, where only a few nodes possess high degrees while most nodes have low degrees. Degree distributions can also be used to distinguish noises from the truth. For example, when spammers connect to users, the degree distribution of the network may be warped. 

Considering the formation of normal edges, nodes tend to form an edge with nodes with a high degree, which enables us to generate graphs following power-law degree distributions, hence generating more realistic true graphs against noises. 
We assign weights $\mathbf{D} = \{d_i\}_n$ to nodes to represent the ``fitness" of nodes to win edges, which can be reflected by the degree of a node $v_i$. 
We assume that $d_i$ follow weighted exponential distribution $d \sim E(\lambda)$, where $\lambda$ is a hyper-parameter.

Consequently, the probability of an edge to be generated in the true graph $\mathbf{A}$ is defined according to the similarity between two nodes and the importance of them. Nodes with greater ``fitness" tend to generate edges, and hence we rewrite $a_{ij}$ in Equation \ref{equ:prob} as 
\begin{equation}
    a_{ij} = \sigma(\langle  \mathbf{u}_i, \mathbf{u}_j \rangle + d_i d_j).
\end{equation}

\subsection{Noise Generator with Adaptive Distributions}
\label{sec:noise_generator}

Generation of noises on graphs has specific patterns and is distinguishable from the generation of normal edges. 
On one hand, the fundamental properties of noises are randomness, independence and sparsity. 
On the other hand, real-world noises on graphs may follow different distributions, which are difficult to obtain in advance. 
In this section, we design the noise generator that follows the noise properties and achieves good adaptability so that it can adapt to real-world noises with richer notions.

We first consider a simple situation where edge noises follow zero-mean Gaussian distribution $\epsilon_{ij} \sim \mathcal{N}(0, \sigma_{\epsilon}^2)$,
where $\sigma_{\epsilon}$ is a hyper-parameter. The noises on most of edges are close to zero, that is, the noises hardly affect these edges. Only a small part of edges are heavily noisy. A large positive noise indicates that the noise is likely to cause an abnormal edge and a negative noise indicates that the noise causes the edge that should have appeared to be missing.

The na\'ive prior distribution mentioned before should be significantly improved to learn real noises with more complex distributions. We use Gaussian distribution as an example and other distributions like the Laplace distribution can be used similarly. Given a series of Gaussian distributions $\mathcal{N}(\mu_l, \sigma_l^2), l=1,..,L$, we assume that real noises follow one of these distributions. We let noises be a linear combination of these distributions
\begin{equation}
\label{equ:noise_multiple}
    \epsilon_{ij} \sim  \sum_l \beta_l \mathcal{N}(\mu_l, \sigma_l^2),
\end{equation}
where $\beta_l$, $\mu_l$ and $\sigma_l^2$ denote the weight, mean and variance of the $l$-th prior distribution. $\mu_l$ and $\sigma_l^2$ are hyper-parameters and can be interpreted as the strength of noises. $\beta_l$ are trainable parameters which are normalized using a softmax layer such that $\sum_l \beta_l = 1$, providing adaptability to handle different noises in real-world graphs.

\subsection{Joint Learning of Two Generators}
\label{sec:learning}

To learn node representations and eliminate noises simultaneously, we derive the joint optimization process of two generators, and present the optimization goals given different prior knowledge, providing solutions for more extensions.

Node representations $\mathbf{U}$ are generated firstly, followed by the generation of the true graph $\mathbf{A}$, which is then combined with noises $\mathbf{E}$ generated independently. In practice, to guarantee $a_{ij} + \epsilon_{ij}$ is still a probability, we truncate its value in $[0,1]$. Finally, our observation $\mathbf{\hat{A}}$ is generated based on the truncated combination of true graph and noises. 
As the observed graph $\hat{\mathbf{A}}$ is binary, with $\hat{a}_{ij} = 1$ if $v_i$ and $v_j$ are connected, we assume that observations on edges $\hat{a}_{ij}$ follow Bernoulli distribution to bridge such gap between continuous and discrete variables: $\hat{a}_{ij}  \sim Bernoulli(a_{ij} + \epsilon_{ij})$. 
Hence, we derive the likelihood function of the whole process as 
\begin{equation}
\begin{aligned}
    P(\mathbf{\hat{A}}|\mathbf{A},\mathbf{E}) P(\mathbf{E}) P(\mathbf{A}|\mathbf{U}) P(\mathbf{U} | \mathbf{Z}) P(\mathbf{Z}),
\end{aligned}
\label{equ:likelihood}
\end{equation}
where $P(\mathbf{A}|\mathbf{U}) P(\mathbf{U} | \mathbf{Z}) P(\mathbf{Z})$ and $P(\mathbf{E})$ correspond to the graph generator and the noise generator respectively. $\mathbf{Z}$ denotes prior knowledge on node representations $\mathbf{U}$. We find that when taking log-likelihood, Equation \ref{equ:likelihood} resembles the objective functions of SkipGram models with regularization terms, hence facilitating the optimization of our model. 

\begin{lemma}
    \label{lem:1}
    Equation \ref{equ:likelihood} is equivalent to the overall loss as:
    \begin{equation}
    \begin{aligned}
        \mathcal{L} = \mathcal{L}_{proximity}(\mathbf{\hat{A}},\mathbf{U},\mathbf{E}) +\alpha_E \mathcal{L}(\mathbf{E}) + \alpha_U \mathcal{L}(\mathbf{U}),
    \end{aligned}
    \end{equation}
    where $\mathcal{L}_{proximity}(\mathbf{\hat{A}},\mathbf{U},\mathbf{E})$ is a modification of the structure-preserving objective with noises added. Regularization terms $\mathcal{L}(\mathbf{E})$ and $\mathcal{L}(\mathbf{U})$ are imposed to constrain noises and node representations respectively with $\alpha_E$ and $\alpha_U$ as their weights. 
\end{lemma}
\begin{proof}
    We prove the lemma on the most na\'ive situation, where only the homophily of graph prior knowledge and the noises with Gaussian distribution are considered. Learning procedures for more extensions will be presented later. We adopt Maximum Likelihood Estimation (MLE) and derive negative log-likelihood function as
    \begin{equation}
        \begin{aligned}
            & - \log P(\mathbf{\hat{A}}|\mathbf{A},\mathbf{E}) P(\mathbf{E}) P(\mathbf{A}|\mathbf{U})  P(\mathbf{U} | \mathbf{Z}) P(\mathbf{Z}) \\
            = & - \log \prod_{i,j \in V} [(a_{ij} + \epsilon_{ij})^{\hat{a}_{ij}} \times (1-a_{ij} - \epsilon_{ij})^{1 - \hat{a}_{ij}} \\
            & \times \phi_{\mu_{\epsilon}, \sigma_{\epsilon}}(\epsilon_{ij}) \times \phi_{0, \sigma_u}(\mathbf{u}_i) \times \phi_{0, \sigma_u}(\mathbf{u}_j)] \\
            = & - \sum_{(i,j) \in E} \log(\sigma(\langle \mathbf{u}_i, \mathbf{u}_j \rangle) + \epsilon_{ij}) \\
            & - k \cdot E_{v_n \sim P_n(v)}[\log(1-\sigma(\langle \mathbf{u}_i, \mathbf{u}_j \rangle) - \epsilon_{ij})] \\
            & + \alpha_E \sum_{i,j \in V} \epsilon_{ij}^2 + \alpha_U \sum_{i \in V} \sum_r u_{i,r}^2 \\
            = & \mathcal{L}_{proximity}(\mathbf{\hat{A}},\mathbf{U},\mathbf{E}) +\alpha_E \mathcal{L}(\mathbf{E}) + \alpha_U \mathcal{L}(\mathbf{U}),
        \end{aligned}
        \label{equ:loss}
        \end{equation}
    where $\phi_{\mu, \sigma}(x) = \frac{1}{\sqrt{2 \pi} \sigma} \exp(- \frac{(x-\mu)^2}{2\sigma^2})$ is the density function for Gaussian distributions and $C$ is a constant. And we use negative sampling to obtain negative samples. 
\end{proof}

The structural-preserving term $\mathcal{L}_{proximity}$ preserves local structure of nodes. In practice, we implement $\mathcal{L}_{proximity}$ using the strategy in DeepWalk, letting $\hat{a}_{ij}$ be a positive sample if $v_j$ is the context of $v_i$ in a random walk. 
$\mathcal{L}(\mathbf{E})$ is imposed to constrain the noises to eliminate the influence of noises from node representations. 
Another regularization term $\mathcal{L}(\mathbf{U})$ is imposed to regularize representations and also enables the model to learn robust representations. 

\subsubsection{Learning of Diverse Graph Prior Knowledge}

Given the community prior knowledge, $\mathcal{L}(\mathbf{U})$ can be replaced as 
\begin{equation}
\begin{aligned}
    \mathcal{L}_{com}(\mathbf{U}) = & \sum_{i \in V} \sum_{r < m_1} u_{i,r}^2 +  \sum_{i \in V} \sum_{k} \frac{\gamma_{ik}}{2 \sigma_c^2} \sum_{r \ge m_1} (u_{i,r}-c_{k,r})^2.
\end{aligned}
\label{equ:community}
\end{equation}
Given the power-law distribution of node degrees, we have 
\begin{equation}
\mathcal{L}_{deg}(\mathbf{U}) = \sum_{i \in V} \lambda d_i.
\end{equation}
Furthermore, different graph priors can be combined flexibly by weighted summation of each regularization term $\mathcal{L}_z(\mathbf{U})$.

\subsubsection{Learning of Adaptive Noise Distributions}

When the noises are generated according to Equation \ref{equ:noise_multiple}, we have
\begin{equation}
\begin{aligned}
     \mathcal{L}_{adap}(\mathbf{E}) = \sum_{i,j \in V} \sum_l \frac{\beta_l}{2 \sigma_l^2} (\epsilon_{ij} - \mu_l)^2.
\end{aligned}
\end{equation}

\begin{table*}[htbp]
    \centering
    \caption{Macro-F1 of Node Classification on Real-World Data.}
    \begin{tabular}{c|c|ccc|ccc|ccc|ccc}
    \toprule
    \multicolumn{2}{c|}{Dataset} & \multicolumn{3}{c|}{Cora} & \multicolumn{3}{c|}{Email} & \multicolumn{3}{c|}{Facebook} & \multicolumn{3}{c}{Blog Catalog} \\ 
    \midrule
    \multicolumn{2}{c|}{\% Training Data} & 30\% & 50\% & 70\% & 30\% & 50\% & 70\% & 30\% & 50\% & 70\% & 30\% & 50\% & 70\% \\
    \midrule
    \multirow{5}{*}{\shortstack{Traditional NRL Models}} & DeepWalk & 	78.12	&	79.74	&	80.72	&	42.55	&	48.89	&	54.23 &	39.82	&	40.97	&	41.04 &	19.95	&	21.04	&	21.57	\\
    & Node2Vec & 	78.56	&	80.02	&	80.70	&	45.48	&	50.66	&	56.07   &	43.97	&	44.20	&	44.38	&	19.89	&	21.08	&	21.72 \\
    & LINE &	53.15	&	55.88	&	57.14	&	27.34	&	34.05	&	40.08 &	38.32	&	38.70	&	39.21	&	17.94	&	19.36	&	20.19	\\
    & GCN   &   68.57	&	69.15	&	69.87	&	34.09	&	37.22	&	38.97	&	29.75	&	29.97	&	30.57	&	3.06	&	3.15	&	3.11    \\
    & GraphSAGE &  76.11	&	76.97	&	77.06	&	26.04	&	30.18	&	30.39	&	25.25	&	27.28	&	30.34	&	3.27	&	3.31	&	3.15	\\
    \midrule
    \multirow{3}{*}{\shortstack{Step-Wise Methods}} & Jaccard-DW & 79.97	&	81.23	&	81.13	&	39.84	&	44.19	&	50.92	&	38.82	&	42.31	&	39.77	&	18.38	&	19.12	&	19.62 \\
    & Oddball-DW & 78.73	&	80.86	&	80.51	&	41.89	&	47.33	&	53.80	&	39.56	&	40.33	&	40.66	&	20.12	&	21.36	&	21.94 \\
    & DW-PCA & 78.27	&	79.56	&	79.05	&	47.14	&	49.88	&	54.42	&	35.85	&	36.49	&	36.49	&	19.44	&	20.23	&	20.64\\
    \midrule
    \multirow{4}{*}{\shortstack{Robust NRL Models}} & AIDW &	78.48	&	80.45	&	81.03	&	48.29	&	52.83	&	57.08 &	38.67	&	38.81	&	38.84	&	20.20	&	21.65	&	22.52	\\
    & ARGA & 75.99	&	77.60	&	77.74	&	35.57	&	40.62	&	45.43	&	32.67	&	33.22	&	32.95	&	4.65	&	4.88	&	4.87 \\
    & BGCN &  70.05   &   70.27   &   71.39   &   33.45	&	37.53	&	41.88 &   34.30   &   34.51   &   37.47   &   4.48    &   4.96    &   5.11 \\
    \cmidrule{2-14}
    & DenNE & \textbf{80.18}	&	\textbf{81.60}	&	\textbf{82.44}	&	\textbf{50.42}	&	\textbf{55.07}	&	\textbf{59.56} &	\textbf{49.23}	&	\textbf{49.76}	&	\textbf{50.21} &	\textbf{24.13}	&	\textbf{25.22}	&	\textbf{25.74}	\\
    \bottomrule
    \end{tabular}
    \label{tab:nc_result_macro}
\end{table*}
        
\section{Experiments}

In this section, we evaluate our model on six datasets against eight state-of-the-art approaches for several tasks.

\subsection{Experimental Setup}

\subsubsection{Datasets} 
\textit{Cora}\footnote{Available at https://linqs.soe.ucsc.edu/data.} is a paper citation network with 2,708 nodes, 5,278 edges and 7 labels. \textit{Email}\footnote{Available at http://snap.stanford.edu/data.} is a social network with 1,005 nodes, 25,571 edges and 42 labels. \textit{Facebook}\footnote{Available at http://networkrepository.com/socfb.php.} is a social network with 9,414 nodes, 425,639 edges, 23 groups and 3 labels. \textit{Blog Catalog}\footnote{Available at http://networkrepository.com/soc-BlogCatalog.php.} is a social network with 10,312 nodes, 333,983 edges and 39 labels. \textit{Geometric} is a synthetic network with 1,024 nodes and 5,470 edges generated by a random geometric graph generator\footnote{Available at https://networkx.github.io.}. \textit{Partition} is a synthetic network with 1,024 nodes, 7,066 edges and 8 groups generated by a random partition graph generator\footnote{Available at https://networkx.github.io.}.

\subsubsection{Baselines}
Traditional NRL models include \textit{DeepWalk} \cite{perozzi2014deepwalk}, \textit{Node2Vec} \cite{grover2016node2vec}, \textit{LINE} \cite{tang2015line}, \textit{GCN} \cite{kipf2016semi} and \textit{GraphSAGE} \cite{hamilton2017inductive}. 
Step-wise methods include \textit{Jaccard-DW} \cite{akoglu2015graph}, \textit{OddBall-DW} \cite{akoglu2010oddball} and \textit{DW-PCA}.
Robust NRL models include \textit{AIDW} \cite{dai2018adversarial}, \textit{ARGA} \cite{pan2018adversarially} and \textit{BGCN} \cite{zhang2019bayesian}. 
Besides, \textit{DenNE} denotes our proposed model. \textit{DenNE-Basic} is the basic version without any additional prior knowledge. \textit{DenNE-Adap} is a variant that learns noises adaptively. \textit{DenNE-Com} and \textit{DenNE-Deg} are variants considering community and power-law priors.

\subsubsection{Parameter Settings}
We set the representation size for all models to 128 on real-world graphs and 32 on synthetic graphs.  
For DeepWalk, Node2Vec and AIDW, we set the walk number to 10, the walk length to 80 and the window size to 5. We set $p,q$ to $0.25,0.75$ for Node2Vec. For AIDW, the number of hidden layers for GAN is 1. 
For GNN models, we adopt the unsupervised loss and the number of layers is 2, where the size for each layer is set to 100 and 128 / 32 respectively. 128-dimension features learned from a pre-trained DeepWalk are used as node features. Other parameters of ARGA and BGCN are set to the default setting. 
In our model, we set $\alpha_N=0.001$, $\alpha_E=50.0$ and $\alpha_z=1$ for each knowledge. Other parameters are the same as DeepWalk. Besides, for DenNE-Adap, $(\mu_l, \sigma_l^2)$ are set to (0,5), (0,0.5) and (0.01, 0.5). For DenNE-Com, $\bm{\gamma}$ are available in Facebook and Partition. Community size is set to 7, 42, 39 and 8 in Cora, Email, Blog Catalog and Geometric. For DenNE-Deg, $\lambda$ is set to 2.

\subsection{Experimental Results}

\subsubsection{Node Classification}
We evaluate the robustness of our model for node classification on four real-world networks. We train a logistic regression classifier from scikit-learn package, using node representations learned by each model. On Email, Cora and Facebook, it is a multi-class classification task, while on Blog Catalog, it is a multi-label classification task. We randomly split the data into training and testing set and the results are averaged over 20 times. 

Macro-F1 are reported in Table \ref{tab:nc_result_macro}. Our model outperforms all other baselines consistently. 
DenNE achieves about 2\% gain on Cora and Email, and more significant gains on Facebook and Blog Catalog, which should be attributed to the ability of our model to learn robust representations. 
Compared with DeepWalk, the step-wise methods are unstable because information loss may occur during the step-by-step optimization process.
AIDW performs better than DeepWalk, and ARGA and BGCN perform better than GCN in most datasets. But their abilities to absorb noises are all limited. Our model has better performance on networks with inherent noises thanks to explicitly handling noises in network space and the use of prior knowledge in our model. 

\begin{figure*}[htbp]
\centering
\begin{minipage}[b]{1\linewidth}
    \subfigure[Positive Ratio on Geometric]{
        \includegraphics[width=0.23\columnwidth]{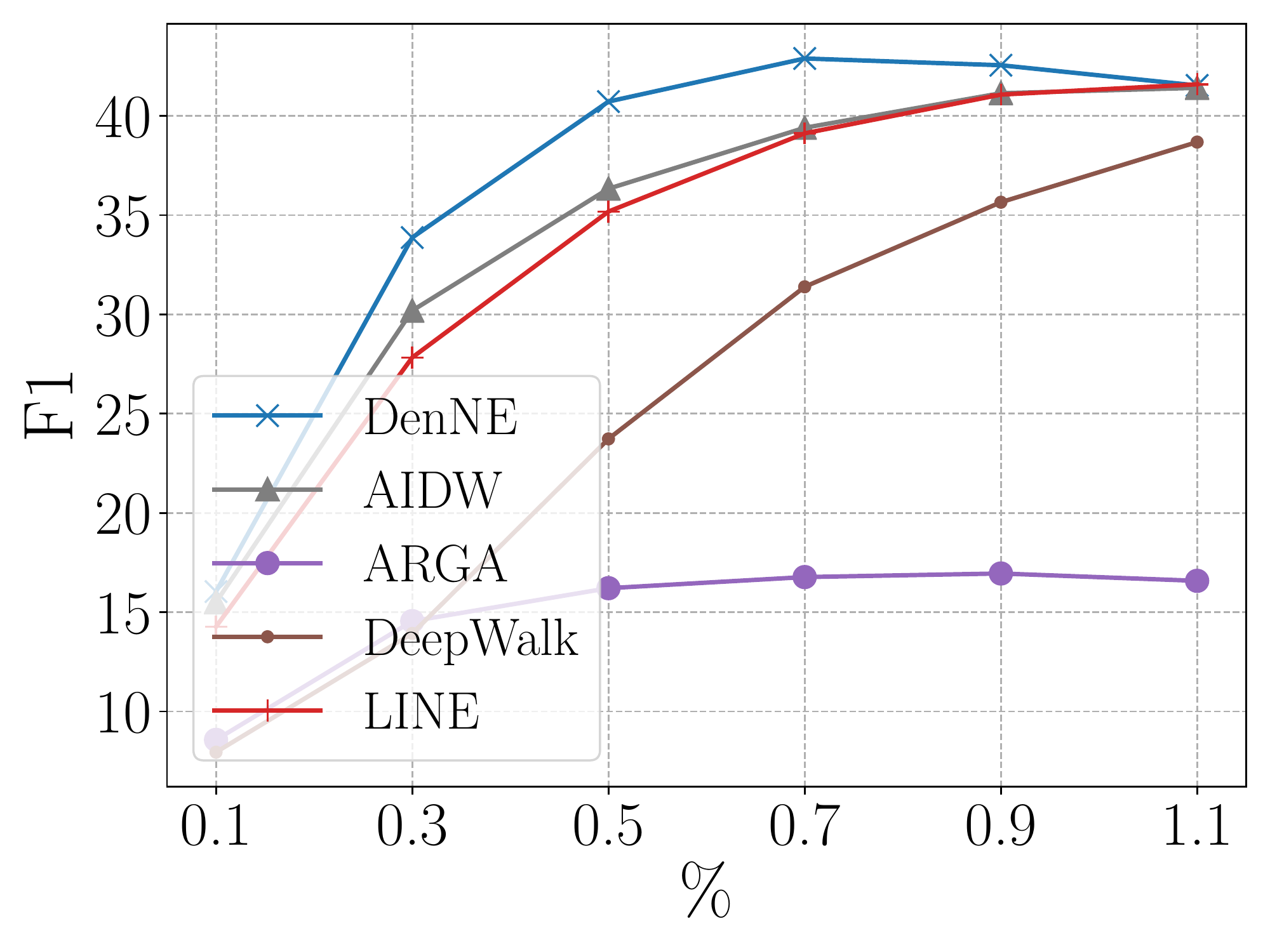} 
        \label{fig:gr_geometric}
    }
    \subfigure[Positive Ratio on Partition]{
        \includegraphics[width=0.23\columnwidth]{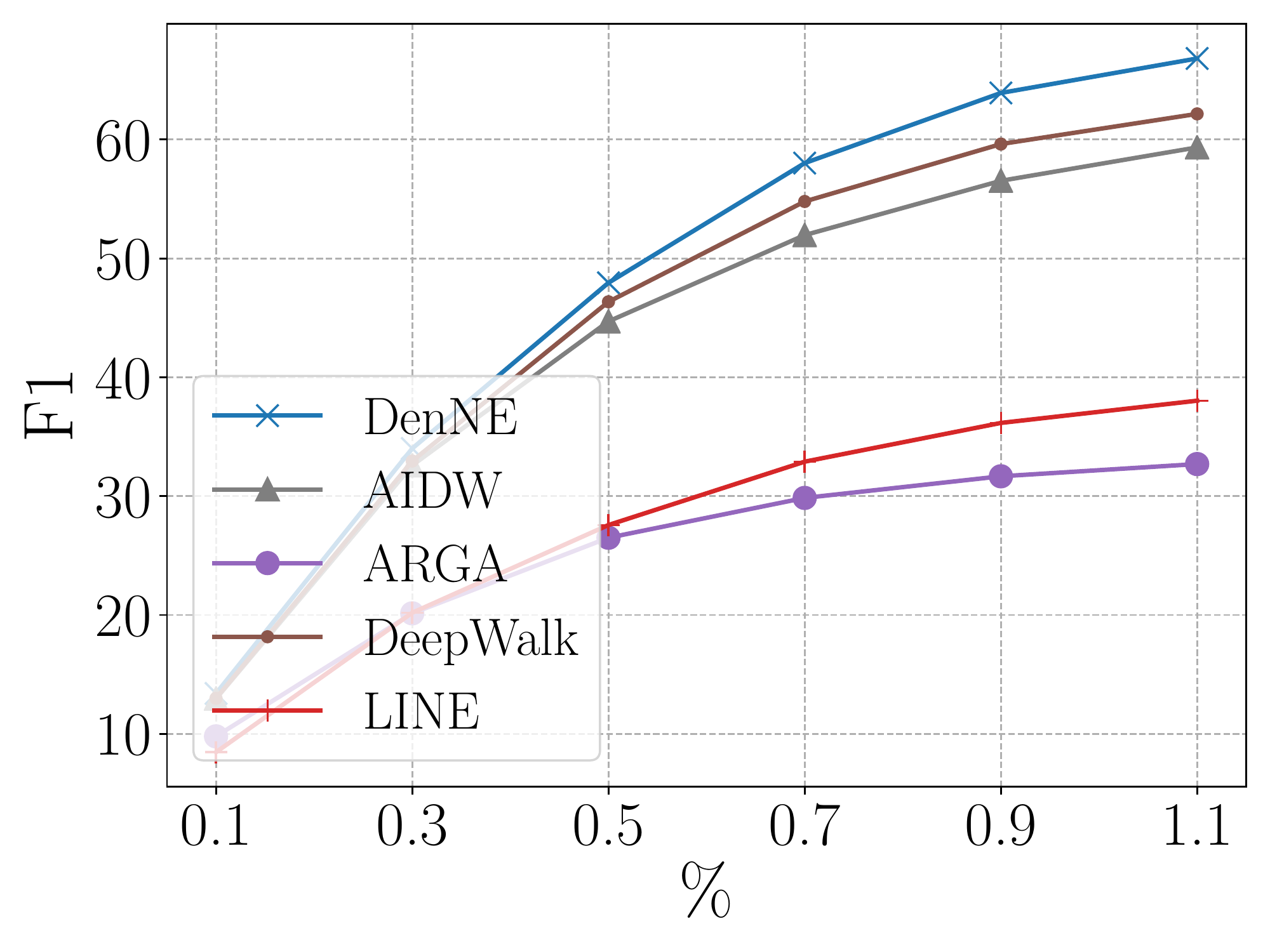} 
        \label{fig:gr_partition}
    }
    \subfigure[Noise Ratios on Partition]{
        \includegraphics[width=0.23\columnwidth]{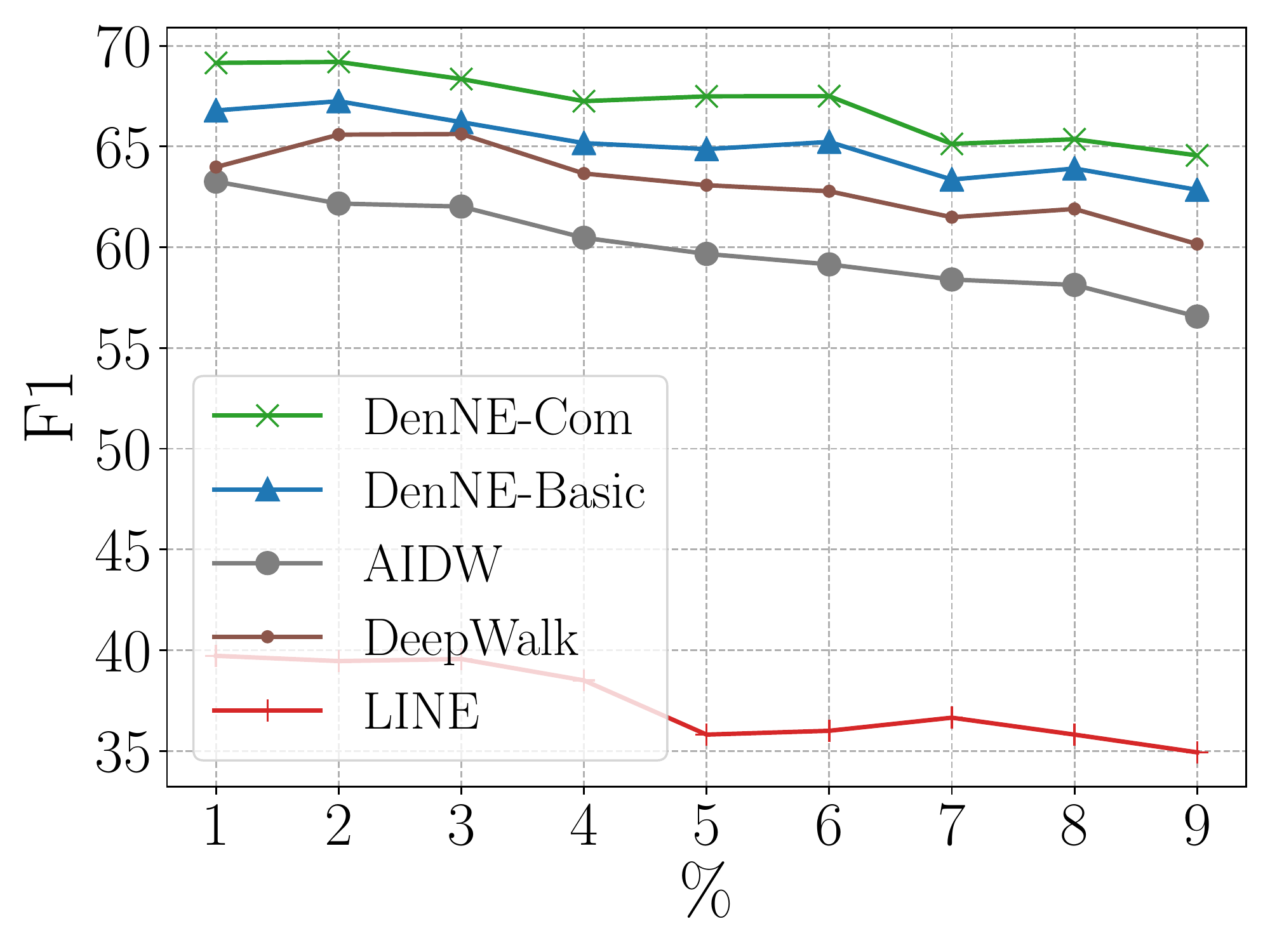} 
        \label{fig:noise_ratio}
    }
    \subfigure[Noise Distributions on Geometric]{
        \includegraphics[width=0.23\columnwidth]{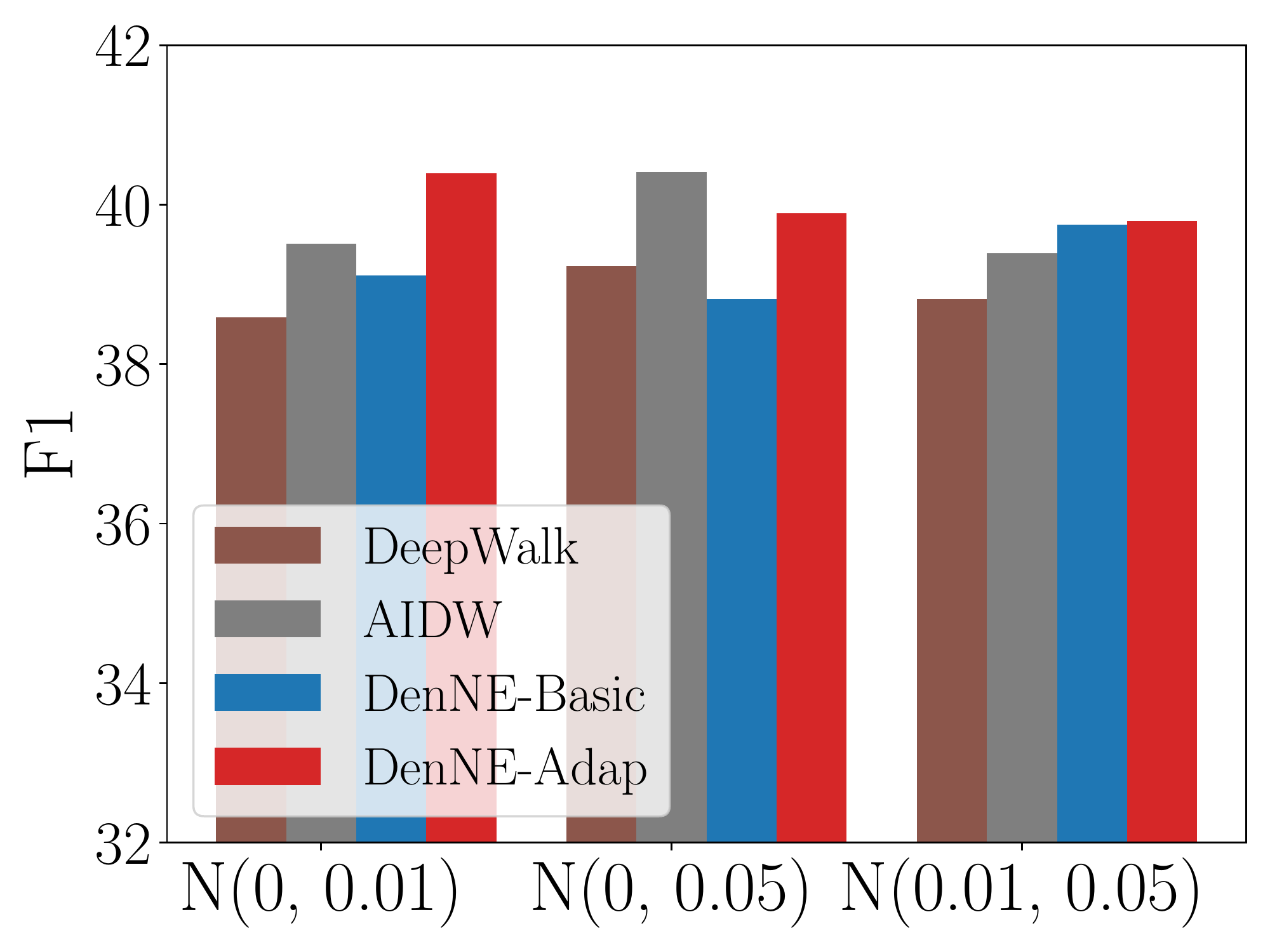} 
        \label{fig:noise_type}
    }
    \caption{Results of Graph Reconstruction on Synthetic Data.} 
    \label{fig:gr}
\end{minipage}
\end{figure*}

\begin{figure*}[htbp]
    \begin{minipage}[b]{1\linewidth}
        \subfigure[$\alpha_U$ on Cora]{
            \includegraphics[width=0.23\columnwidth]{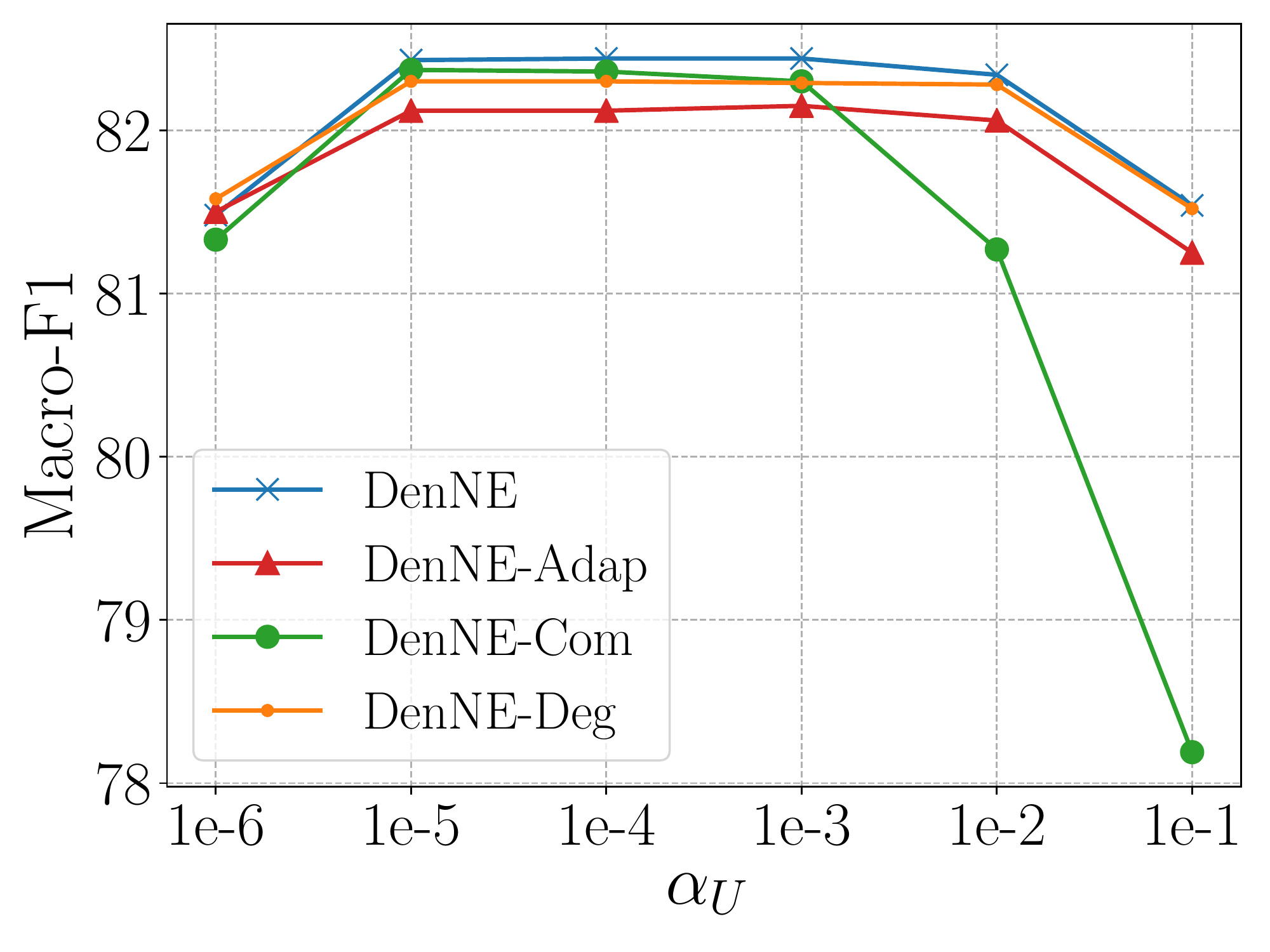} 
            \label{fig:par_embed}
        }
        \subfigure[$\alpha_E$ on Cora]{
            \includegraphics[width=0.23\columnwidth]{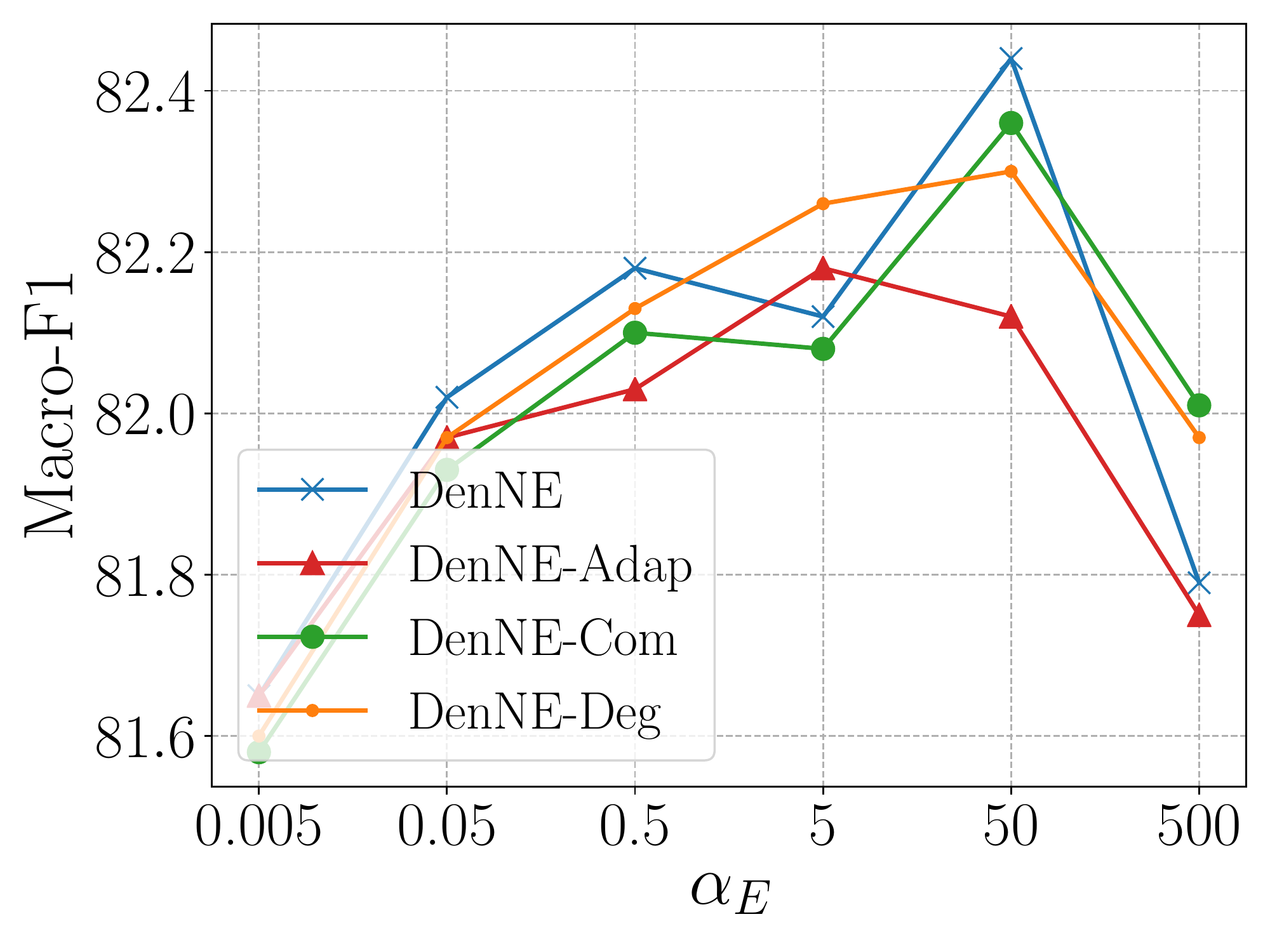} 
            \label{fig:par_noise}
        }
        \subfigure[Prior Knowledge on Geometric]{
            \includegraphics[width=0.23\columnwidth]{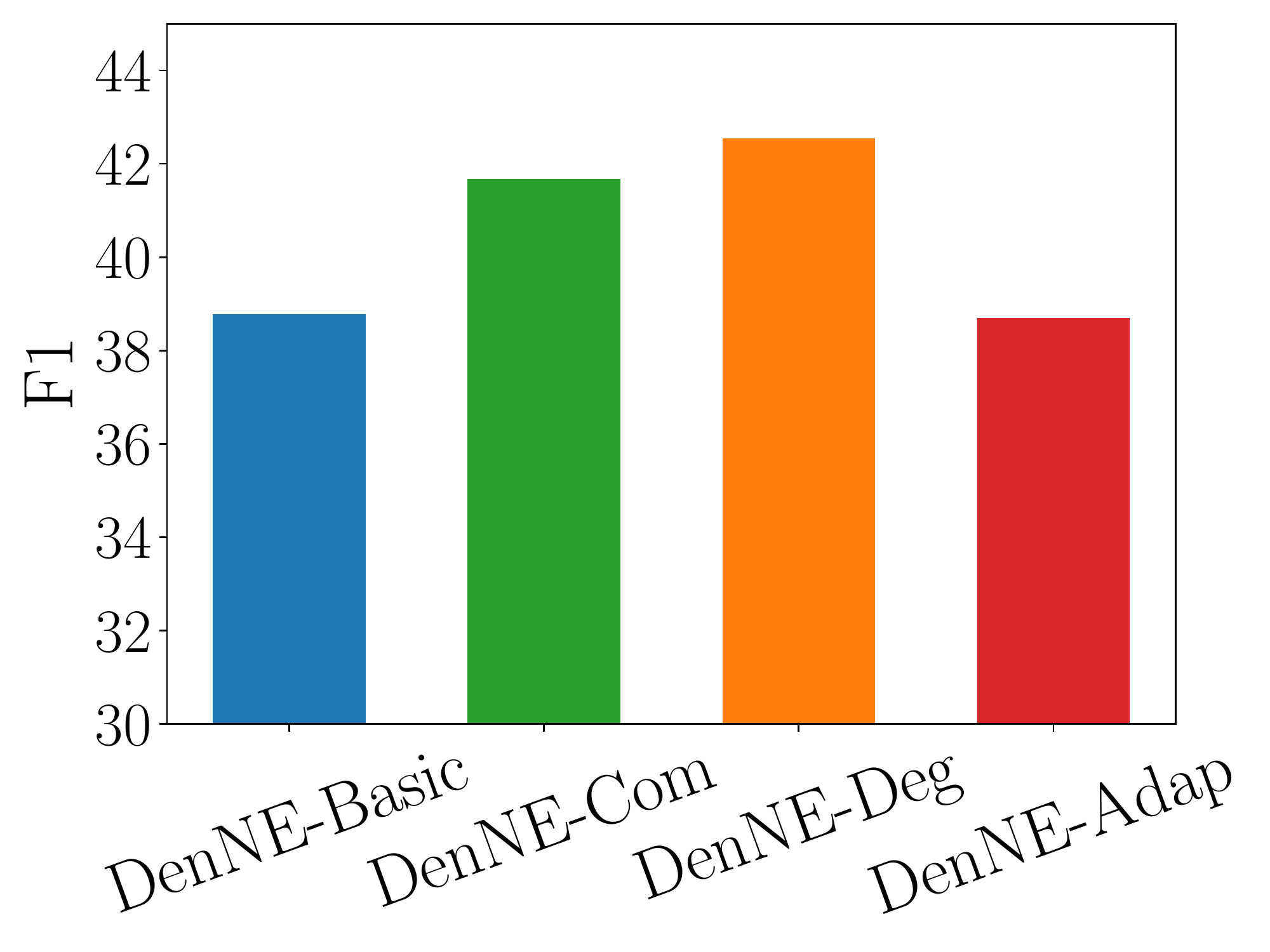} 
            \label{fig:par_geometric}
        }
        \subfigure[Prior Knowledge on Partition]{
            \includegraphics[width=0.23\columnwidth]{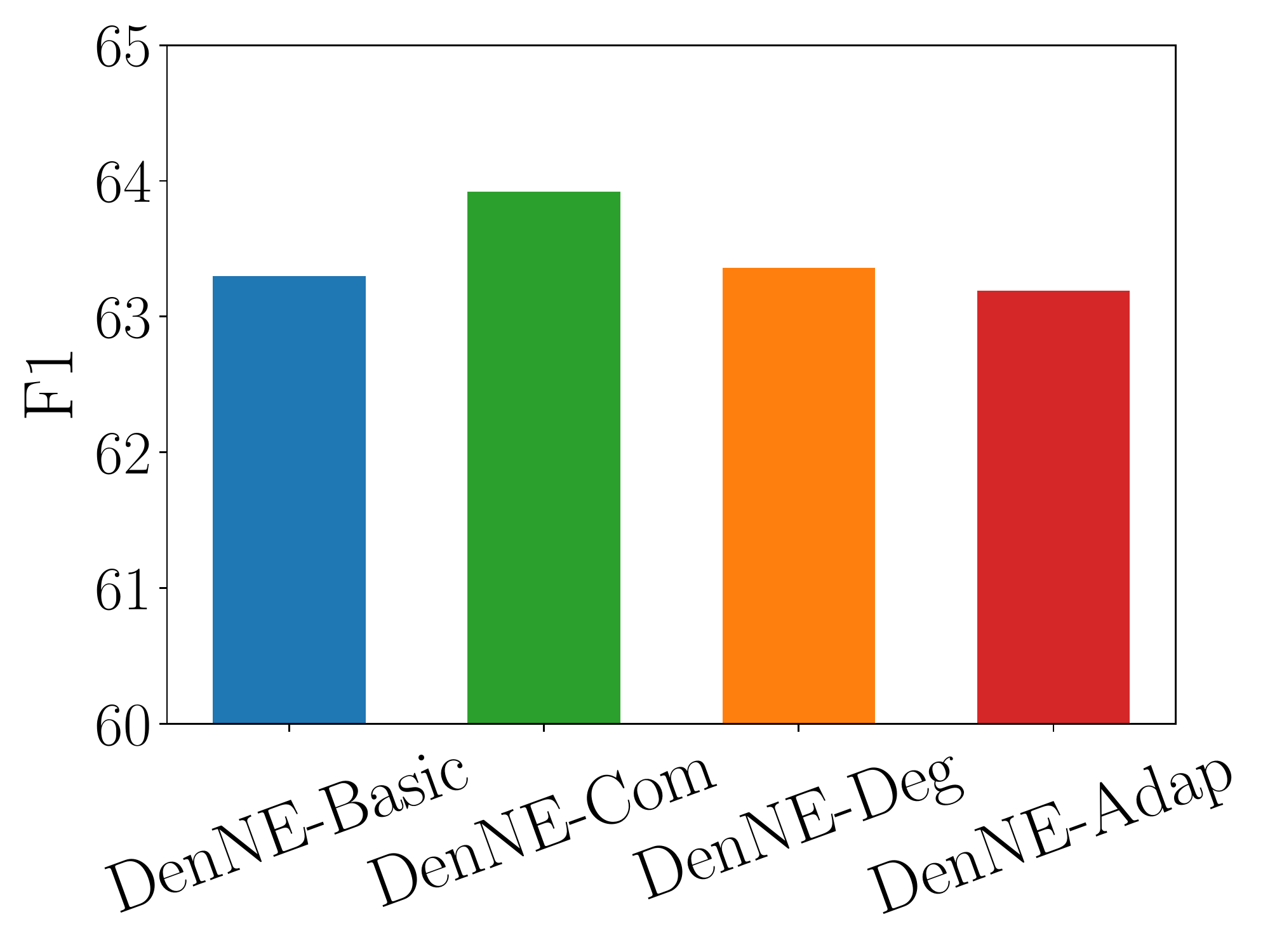} 
            \label{fig:par_partition}
        }
        \caption{Analysis of Parameter Sensitivity.} 
        \label{fig:par}
    \end{minipage}
    \end{figure*}

\subsubsection{Graph Reconstruction}

We reconstruct the original pristine graphs using representations learned from the noisy graphs. Two synthetic graphs generated according to different principles are used as pristine graphs, and a certain ratio of edges are added or removed. Euclidean distance between representations of each node pair are computed as $\|\mathbf{u}_i-\mathbf{u}_j\|_2^2$. Then node pairs are sorted based on distances and the smallest $k$ samples are regarded as positive samples.

F1 score of graph reconstruction when 5\% of edges are added as noises are shown in Figure \ref{fig:gr_geometric} and Figure \ref{fig:gr_partition}. We range the ratio of positive pairs to all pairs from 0.1\% to 1.1\%. Our model has better performance to reconstruct the pristine graphs than both traditional and adversarial embedding models. 
Then we fix the ratio of positive pairs to 1.0\%.
Figure \ref{fig:noise_ratio} shows the result of different ratios of noises added to Partition. The performance of all models become consistently worse if more noises are added, while the results of our model keep better than baselines. 
Figure \ref{fig:noise_type} shows the result of different distributions of noises are added to Geometric, indicating that learning noises adaptively performs better than given a simple prior distribution.

\subsubsection{Parameter Sensitivity}

We first take node classification on Cora as an example to analyze $\alpha_U$ and $\alpha_E$. As shown in Figure \ref{fig:par_embed}, when $\alpha_E$ is fixed to 50.0, the performance of our model is stable unless $\alpha_U$ is too larger or too smaller. Figure \ref{fig:par_noise} shows that when $\alpha_U$ is fixed to 0.0001, the performance will be degraded if $\alpha_E$ is too strong or weak. 
Then the effectiveness of prior knowledge is evaluated on synthetic datasets. On Geometric dataset, the graph have power-law property so it will have much improvement if the degree prior is considered as shown in Figure \ref{fig:par_geometric}. On Partition dataset with obvious community structure, DenNE-Com enjoys the most improvements in Figure \ref{fig:par_partition} .

\section{Conclusion}

We discuss how to learn noise-free node representations on corrupted graphs. In the future, we will expand the model to more general networks such as temporal networks, and to deep models which requires further theoretical analysis.


\bibliography{IEEEtran}
\bibliographystyle{IEEEtran}



%



\end{document}